\documentclass[10pt,journal,onecolumn]{IEEEtran}
\usepackage{amsfonts}
\usepackage{amsmath}
\usepackage{amssymb}
\usepackage{bm}
\usepackage{xcolor,graphicx,float}
\usepackage{color, soul}
\usepackage{stfloats}
\usepackage[numbers,sort&compress]{natbib}
\usepackage[amsmath,thmmarks]{ntheorem}
\usepackage{theorem}
\usepackage{algorithm}
\usepackage{algorithmic}
\usepackage{graphicx}
\usepackage{subfigure}
\usepackage{pict2e}

\newtheorem{proposition}{Proposition}

\theoremheaderfont{\sc}\theorembodyfont{\upshape}
\theoremstyle{nonumberplain}
\theoremseparator{}
\theoremsymbol{\rule{1ex}{1ex}}
\newtheorem{proof}{Proof}

\begin{document}
\title{Asymptotically Optimal One-Bit Quantizer Design for Weak-signal Detection in Generalized Gaussian Noise and Lossy Binary Communication Channel}

\author{Guanyu Wang, Jiang~Zhu and Zhiwei~Xu
}
\maketitle
\begin{abstract}
In this paper, quantizer design for weak-signal detection under arbitrary binary channel in generalized Gaussian noise is studied. Since the performances of the generalized likelihood ratio test (GLRT) and Rao test are asymptotically characterized by the noncentral chi-squared probability density function (PDF), the threshold design problem can be formulated as a noncentrality parameter maximization problem. The theoretical property of the noncentrality parameter with respect to the threshold is investigated, and the optimal threshold is shown to be found in polynomial time with appropriate numerical algorithm and proper initializations. In certain cases, the optimal threshold is proved to be zero. Finally, numerical experiments are conducted to substantiate the theoretical analysis.
\end{abstract}
{\bf Keywords:} Threshold optimization, weak-signal detection, quantization, generalized Gaussian noise

\section{Introduction}
Signal estimation and detection from quantized data continues to attract attention over the past years \cite{Poor1, Li1, Ribeiro3, Pan1, Ciuonzo2, Jiang1, Jiang2, Sani, Jiang3,Jiang4, Farias1,Farias2, Ciuonzo1,ZJF1,ZJF2,ZJF3}. In \cite{Poor1}, a general result is developed and applied to obtain specific asymptotic expressions for the performance loss under uniform data quantization in several signal detection and estimation problems including minimum mean-squared error (MMSE) estimation, non-random point estimation, and binary signal detection. In \cite{Li1}, a distributed adaptive quantization scheme is proposed for signal estimation, where individual sensor nodes dynamically adjusts their quantizer threshold based on earlier transmissions from other sensor nodes. In \cite{Ribeiro3}, distributed parameter estimators based on binary observations along with their error-variance performance are derived in the case of an unknown noise probability density function (PDF). For the robust estimation of a location parameter, the noise benefits to maximum likelihood type estimators are investigated \cite{Pan1}. As a result, the analysis of stochastic resonance effects is extended for noise-enhanced signal and information processing. In \cite{Ciuonzo2}, distributed detection of a non-cooperative target is tackled, and fusion rules are developed based on the locally-optimum detection framework. Recently, some variants of the classical signal estimation and detection model from quantized data are studied. One is that the unquantized observations are corrupted by combined multiplicative and additive Gaussian noise \cite{Jiang1,Jiang2,Sani}. Another is called the unlabeled sensing where the unknown order of the quantized measurements causes the entanglement of desired parameter and nuisance permutation matrix \cite{Jiang3, Jiang4}. In \cite{Farias1,Farias2}, the authors investigate the estimation problem under generalized Gaussian noise (GGN) and reveal the property of the Fisher information (FI). In addition, a systematic framework for composite hypothesis testing from independent Bernoulli samples is studied, and the comparison of detectors are made under one-sided and two-sided assumptions \cite{Ciuonzo1}.

The threshold of the quantizer can be designed to improve the performance of estimation and detection \cite{Kassam1, Warren1, Willett3, Chen1,Venkitasubramaniam1, Junfang1,Rousseau1, Ciuonzo3, Ciuonzo4, Ciuonzo5}. In the early paper \cite{Kassam1}, two useful detection criteria are proposed, leading to the MMSE between the quantized output and the locally optimum nonlinear transform for each data sample. Later in \cite{Warren1}, the optimal quantized  detection problem is considered for the Neyman-Pearson, Bayes, Ali-Silvey distance, and mutual (Shannon) information criteria, and it is shown that the optimal sensor decision rules quantize the likelihood ratio of the observations. In the design of quantized detection systems, the optimal test is shown to employ a nonrandomized rule under certain conditions, which considerably simplifies the design \cite{Willett3}. In \cite{Chen1}, it is shown that given a particular constraint on the fusion rule, the optimal local decisions which minimize the error probability amount to a likelihood-ratio test (LRT). In addition, a design example with a binary symmetric channel (BSC) is given to illustrate the usefulness of the result in obtaining optimal threshold for local sensor observations. In \cite{Venkitasubramaniam1}, the maximin asymptotic relative efficiency (ARE) criterion is proposed to optimize the thresholds, and the improvement of estimation performance is demonstrated in distributed systems. Utilizing the asymptotic performance of the one-bit generalized likelihood ratio test (GLRT) detector, the optimal threshold is proven to be zero under Gaussian noise and a BSC \cite{Junfang1}. The quantizer design is also analyzed under the GGN \cite{Rousseau1,Ciuonzo3}. In \cite{Rousseau1}, the problem is  considered under the error-free channel, and the optimal threshold is only plotted without theoretical justification. The BSC is also included in the successional studies \cite{Ciuonzo3,Ciuonzo4,Ciuonzo5}. In \cite{Ciuonzo3}, zero is shown to be the optimal threshold when the shape coefficient is less than or equal to two and a good (sub-optimal) choice when the shape coefficient is larger than two. Analogously, zero is employed as a good choice in the generalized Rao test \cite{Ciuonzo4}. For generalized locally optimum detectors, the threshold optimization is re-formulated as a maximization problem in terms of the local false-alarm probability, which can be easily evaluated via one-dimension numerical search \cite{Ciuonzo5}.

\subsection{Related Work and Main Contributions}
The most related work to ours is \cite{Junfang1, Rousseau1, Ciuonzo3}. Compared to \cite{Junfang1} focusing on Gaussian noise only in the BSC setting, we study the threshold optimization problem under GGN and arbitrary binary channel. In \cite{Rousseau1}, the authors present the optimal threshold without theoretical proof, and they do not take the binary channel into account. In \cite{Ciuonzo3}, it states that choosing zero threshold is a suboptimal choice (not too bad). In this paper, we extend their work to more general settings. It should be noticed that compared to the wide use of Gaussian noise assumption, the GGN assumption is usually made for infrequent but high level events, e.g., extremely low frequency electromagnetic noise due to thenderstorms or under ice acoustic noise due to iceberg break. In these events, the GGN assumption models the noise spikes more accurately than the Gaussian one and thus leads to better detection performance \cite[p.381]{Kay2}.

The main contribution of this paper is to address the threshold design problem under GGN in the arbitrary binary channel. Under the weak-signal assumption, the thresholds can be optimized via maximizing the noncentrality parameter. Unfortunately, it is difficult to prove the  theoretical properties of the noncentrality parameter function with respect to the threshold. We novelly propose a simplified function whose sign is the same as that of the first derivative of the noncentrality parameter function. Consequently, we rigorously prove the theoretical properties of the  noncentrality parameter function with respect to the threshold indirectly. Then we prove that for arbitrary binary channel, the optimal threshold can be found in polynomial time via appropriate numerical algorithm with proper initializations. In certain cases, we prove the optimal threshold to be zero.
\subsection{Organization}
The paper is organized as follows. In section \ref{section_problem}, the weak-signal detection problem is described, and preliminary materials including both maximum likelihood (ML) estimation and parameter tests are introduced. Section \ref{Q_design} states the main results of the quantizer design. In addition, an algorithm to calculate the optimal threshold is proposed. In section \ref{simulation}, numerical experiments are conducted to substantiate the theoretical analysis. The conclusions are presented in section \ref{con}. Finally, the related functions and the proof of propositions are presented in \ref{appendix}.
\section{Problem Setup}\label{section_problem}
In this section, the weak-signal detection problem from binary samples is described. In addition, the ML estimation, GLRT and Rao test are presented.

Consider a binary hypothesis testing problem, in which $N$ distributed sensors in a wireless sensor network (WSN) are utilized $K$ times to generate noisy observations. Those observations are quantized with different thresholds, and then used to detect the presence of an unknown deterministic weak signal with amplitude $\theta$. The quantized samples under both hypotheses are
\begin{equation}\label{hypothesis_testing}
\begin{aligned}
\begin{cases}
&{\mathcal H}_0:b_{ij}={\rm 1}\{w_{ij}\geq \tau_{ij}\},\\
&{\mathcal H}_1:b_{ij}={\rm 1}\{h_{ij}{\theta}+w_{ij}\geq \tau_{ij}\},
\end{cases}
\end{aligned}
\end{equation}
where $i=1,\cdots,N$ denotes the sensor number, $j=1,\cdots,K$ denotes the observation time, $h_{ij}$ is a spatial-temporal signal, $w_{ij}$ is the independent and identically distributed (i.i.d) noise, $\tau_{ij}$ is the threshold of the $i$ th sensor at the $j$ th observation time, and ${\rm 1}\{\cdot\}$ is an indicator function which produces 1 if the argument is true and 0 otherwise. We assume that $\theta\in [-\Delta,\Delta]$ for technical reasons \cite{Papadopoulos1}, where $\Delta$ is a known constant.

Between the quantized data and the fusion center (FC) during the transmission, $b_{ij}$ are flipped to $u_{ij}$ before being received \cite{Ozdemir1}. Let $(q_0,q_1)$ denote the flipping probabilities such that
\begin{equation}\label{binary_channel}
\begin{aligned}
&{\rm Pr}(u_{ij}=1|b_{ij}=0)=q_0,\\
&{\rm Pr}(u_{ij}=0|b_{ij}=1)=q_1,
\end{aligned}
\end{equation}
which will be used to calculate the probability mass function of $u_{ij}$ later.

In this paper, we focus on the asymptotically optimal quantizer design in the case of arbitrary binary channel and the GGN $w$, whose cumulative distribution function (CDF) is $F(w)$ and PDF is
\begin{align}\label{noise_pdf}
f(w)=\frac{\alpha\beta}{2\Gamma(1/\beta)}e^{-(\alpha|w|)^\beta},
\end{align}
where $\Gamma(\cdot), \alpha^{-1}>0, \beta>0$ denote the gamma function, its scale parameter and shape parameter. Note that the GGN can describe several common PDFs such as Laplace distribution ($\beta=1$), Gaussian distribution ($\beta=2$) and uniform distribution ($\beta=\infty$).
\subsection{Maximum Likelihood Estimation}\label{mle_pre}
Under hypothesis $\mathcal H_1$, the PMF of $b_{ij}$ derived from (\ref{hypothesis_testing}) is
\begin{equation}\label{y_probability}
\begin{aligned}
&{\rm Pr}(b_{ij}=1|{\mathcal H_1})=F(h_{ij}\theta-\tau_{ij}),\\
&{\rm Pr}(b_{ij}=0|{\mathcal H_1})=1-F(h_{ij}\theta-\tau_{ij}),
\end{aligned}
\end{equation}
where $F(\cdot)$ denotes the CDF of the GGN $w$. The binary data $b_{ij}$ are transmitted to the FC through the binary channel (\ref{binary_channel}). As a consequence, the PMF of $u_{ij}$ under hypothesis $\mathcal H_1$ can be formulated as
\begin{equation}\label{y_probability}
\begin{aligned}
{\rm Pr}(u_{ij}=1|{\mathcal H_1})&=q_0+(1-q_0-q_1)F\left({h_{ij}{\theta}-\tau_{ij}}\right)\triangleq p_{ij},\\
{\rm Pr}(u_{ij}=0|{\mathcal H_1})&=1-p_{ij}.\\
\end{aligned}
\end{equation}
Let $\mathbf U$ be the matrix satisfying $[\mathbf U]_{ij}=u_{ij}$. The PMF of $\mathbf U$ under hypothesis $\mathcal H_1$ is
\begin{align}\label{LK}
p(\mathbf U;{\theta}|{\mathcal H_1})=\prod_{i=1}^{N} \prod_{j=1}^{K} {\rm Pr}(u_{ij}=1|{\mathcal H_1})^{u_{ij}} {\rm Pr}(u_{ij}=0|{\mathcal H_1})^{(1-u_{ij})},
\end{align}
and the corresponding log-likelihood function $l(\mathbf U;{\theta})\triangleq l(\mathbf U;{\theta}|{\mathcal H_1})$ is
\begin{align}\label{log_likelihood_function}
l(\mathbf U;{\theta})=\sum_{i=1}^{N}\sum_{j=1}^{K} (u_{ij} \log p_{ij}+(1-u_{ij})\log(1-p_{ij})).
\end{align}
Similarly, the log-likelihood under hypothesis $\mathcal H_0$ is $l(\mathbf U|{\mathcal H_0})=l(\mathbf U;0)$.
\subsection{Parameter Tests}
\subsubsection{GLRT}
In the case of known $\theta$, the optimal detector according to the NP criterion is the log-likelihood ratio test \cite[p. 65, Theorem 3.1]{Kay1}. For unknown $\theta$, the GLRT is usually adopted. Although there is no optimality associated with the GLRT, it appears to work well in many scenarios of practical interest \cite[p. 200]{Kay2}. The GLRT replaces the unknown parameter by its ML estimation and decides ${\mathcal H}_1$ if
\begin{equation}\label{GLRT_Detector}
\begin{aligned}
T_G(\mathbf U)=\underset{\theta\in[-\Delta,\Delta]} {\operatorname{max}}~ l({\mathbf U};{\theta})-{ l({\mathbf U};0)}>\gamma,
\end{aligned}
\end{equation}
where $\gamma$ is a threshold determined by the given false alarm probability $P_{FA}$.
\subsubsection{Rao Test}
Since the Rao test does not require an ML estimate evaluation, it is easier to compute in practice \cite[p. 187]{Kay2}. The Rao test decides ${\mathcal H}_1$ if
\begin{equation}\label{GLRT_Detector}
\begin{aligned}
T_R(\mathbf U)=\left(\frac{\partial l({\mathbf U};\theta)}{\partial\theta}\bigg|_{\theta=0}\right)^2 I^{-1}(0)>\gamma,
\end{aligned}
\end{equation}
where $I(0)$ is the FI $I(\theta)$ evaluated at $\theta=0$, and the concrete expression of $I(\theta)$ is presented later in equation (\ref{FI}).
\section{Quantizer Design }\label{Q_design}
In this section, the threshold optimization problem of maximizing the noncentrality parameter is formulated, and the theoretical property of the noncentrality parameter function is revealed. In addition, a gradient descent algorithm is proposed to find the optimal thresholds.

The detection performance of GLRT $T_G(\mathbf U)$ or Rao test $T_R(\mathbf U)$ is difficult to analyze. Fortunately, an approximation can be utilized and it reveals that as $NK\to\infty$, the asymptotic performance of $2T_G(\mathbf U)$ and $2T_R(\mathbf U)$ is \cite[pp. 188-189]{Kay2}
\begin{align}\label{Chi_Square}
2T_G(\mathbf U),2T_R(\mathbf U) \sim
\begin{cases}
&{\mathcal H}_0: \quad \chi_1^2   \\
&{\mathcal H}_1:\quad \chi_1'^{2}(\lambda_Q),
\end{cases}
\end{align}
where $\chi_n^2$ denotes a central chi-squared PDF with $n$ degrees of freedom, and $\chi_n'^{2}(\lambda_Q)$ denotes a noncentral chi-squared PDF with $n$ degrees of freedom and noncentral parameter $\lambda_Q$. In our problem, $\lambda_Q$ is
\begin{equation}\label{lambda_Q}
\begin{aligned}
\lambda_Q = {\theta}^2 I(\theta),
\end{aligned}
\end{equation}
where $I(\theta)$ denotes the FI \cite{Junfang1}. The FI $I(\theta)$ is the expectation of the second derivative of the negative log-likelihood function $l(\mathbf U;{\theta})$ (\ref{log_likelihood_function}) taken w.r.t. $\theta$, i.e.,
\begin{align}
&I({\theta})=(q_0+q_1-1)\sum_{i=1}^{N}\sum_{j=1}^{K}h_{ij}\Bigg[{\rm E}_{\mathbf U}\left(\frac{\partial}{\partial \theta} \left(\frac{u_{ij}}{p_{ij}}-\frac{1-u_{ij}}{1-p_{ij}}\right)\right)\notag\\
&\times f\left(h_{ij}\theta-\tau_{ij}\right)+\frac{\partial}{\partial \theta}f\left({h_{ij}\theta-\tau_{ij}}\right) {\rm E}_{\mathbf U}\left(\frac{u_{ij}}{p_{ij}}-\frac{1-u_{ij}}{1-p_{ij}}\right)\Bigg]\notag\\
&={(1-q_0-q_1)^2}\sum_{i=1}^{N}\sum_{j=1}^{K}\frac{ h_{ij}^2 f^2\left({h_{ij}{\theta}-\tau_{ij}}\right)} {p_{ij}(1-p_{ij})},\label{FI}
\end{align}
where (\ref{FI}) follows due to ${\rm E}_{\mathbf U}[u_{ij}/p_{ij}-(1-u_{ij})/(1-p_{ij})]=0$ and the PMF of $\mathbf U$  (\ref{y_probability}). Under the weak-signal assumption, the unknown scaling $\theta$ takes values near $0$ (actually, we assume that  $|\theta|=c/\sqrt{NK}$ for some constant $c>0$), and we have \begin{align}\label{lambda_app}
\lambda_Q \approx{\theta}^2 I(0)
\end{align}
as $NK\to\infty$ \cite[p. 232]{Kay2}. From (\ref{FI}), we have
\begin{align}
I(0)=(1-q_0-q_1)^2\sum_{i=1}^{N} \sum_{j=1}^{K}h_{ij}^2 G(-\tau_{ij}),\label{FI0}
\end{align}
where
\begin{equation}\label{Gx}
\begin{aligned}
G(x)\triangleq G(x,q_0,q_1)=\frac{f^2(x)}{\frac{1}{4}-\left[(1-q_0-q_1)F(x)-\frac{1}{2}+q_0\right]^2}.
\end{aligned}
\end{equation}
Asymptotically, the noncentrality parameter $\lambda_Q$ determines the detection performance \cite{Junfang1}. Therefore, maximizing the  noncentrality parameter $\lambda_Q$ (\ref{lambda_app}) with respect to $\tau_{ij}$ can be decomposed into a set of independent quantization threshold design problems
\begin{align}\label{optmal_tau}
\tau_{ij}^*=\underset{\tau_{ij}}{\operatorname{argmax}}~ h_{ij}^2G(-\tau_{ij})=\underset{\tau}{\operatorname{argmax}}~ G(-\tau)\triangleq\tau^*.
\end{align}
Equation (\ref{optmal_tau}) demonstrates that the asymptotically optimal weak-signal detection performance can be achieved by utilizing the identical optimal thresholds $\tau^*$, irrespective of the shape of the  spatial-temporal signal, which is also shown in \cite{Junfang1}. The optimal threshold $\tau^*$ can be found via solving the problem
\begin{align}\label{optmal_x}
x^*=\underset{x}{\operatorname{argmax}}~ G(x),
\end{align}
and $\tau^*=-x^*$.

However, $G(x)$ can also be regarded as a function of parameters $q_0$, $q_1$ and $\beta$. Varying these parameters may result in different optimal value $x^*$, more intuitively, different shape of $G(x)$ \cite[Fig.1]{Ciuonzo1}. For better investigation of the theoretical property of $G(x)$, we partition the parameter values and discuss them separately. First, the binary asymmetric channel case which corresponds to $q_0\not=q_1$ is considered. In this setup, the monotonicity or quasiconcavity of $G(x)$ is studied under $0<\beta\leq 1$, $1<\beta\leq 2$ and $\beta>2$ respectively. Second, the deduction in the binary asymmetric channel case is extended to the simple BSC case corresponding to $q_0=q_1$. Finally, combining both cases, we proposed a numerical algorithm to efficiently calculate the optimal value $x^*$ (\ref{Gx}) for arbitrary binary channel. 
\subsection{Binary Asymmetric Channel}\label{BAC}
In this subsection, we focus on the binary asymmetric channel case, i.e., $q_0\not=q_1$. By observing the formula (\ref{Gx}), we realize that a swap of the values of $q_0$ and $q_1$ does not change the value of $G(x)$, namely, $q_0$ and $q_1$ contribute equally to the value of $G(x)$. Heuristically, $G(x)$ may be similar to symmetric functions which have some ``symmetry'' properties, from which Proposition \ref{prop_1} is derived.
\begin{proposition}\label{prop_1}
The maximum of $G(x)$ under arbitrary $q_0\not=q_1$ can be found via solving the problem by restricting $q_0>q_1$ and $1-q_0-q_1>0$, in which the optimal point $x^*\geq 0$.
\end{proposition}
\begin{proof}
The proof is postponed to \ref{prop1}.
\end{proof}
According to Proposition \ref{prop_1}, the problem (\ref{optmal_x}) can be reduced without loss of generality. Hereinafter, we only have to investigate the property of $G(x)$ under $q_0>q_1$, $1-q_0-q_1>0$ and $x>0$. Since the property of a function depends a lot on the derivative, it is necessary to provide the derivative of $G(x)$ as
\begin{align}\label{G_derivative}
G'(x)=\frac{f^3(x)\left[F(x)+\frac{2q_0-1}{2(1-q_0-q_1)}+m_1(x)+m_2(x)\right]M(x)}{(1-q_0-q_1)^2m_1(x)\left[\frac{1}{4(1-q_0-q_1)^2}-\left(F(x)+\frac{2q_0-1}{2(1-q_0-q_1)}\right)^2\right]^2},
\end{align}
where
\begin{subequations}
\begin{align}
&M(x)=F(x)+\frac{2q_0-1}{2(1-q_0-q_1)}+m_1(x)-m_2(x),\label{M_def}\\
&m_1(x)=\frac{f(x)}{2\alpha^{\beta}\beta x^{\beta-1}},\label{rdef}\\
&m_2(x)=\sqrt{\frac{1}{4(1-q_0-q_1)^2}+m_1^2(x)}.\label{sdef}
\end{align}
\end{subequations}
Please notice that $m_1(x)$ and $m_2(x)$ are introduced for compact representation and will be repeatedly used in the following deduction. It is obvious that $m_2(x)>m_1(x)>0$ and $F(x)+\frac{2q_0-1}{2(1-q_0-q_1)}>\frac{q_0-q_1}{2(1-q_0-q_1)}>0$ (due to $x > 0$ and $F(x)>\frac{1}{2}$). Then all the components on the right side of (\ref{G_derivative}) are positive except for $M(x)$. In other words, the sign of $G'(x)$ is the same as that of $M(x)$.
In the following, we deduce the the property of $G(x)$ from $M(x)$, more explicitly, the monotonicity or quasiconcavity of $G(x)$ via $M(x)$ and the derivatives of $M(x)$.
\subsubsection{$0<\beta\leq1$ }\label{beta_01}
In this setup, we prove that $G(x)$ is monotonically decreasing in $(0,+\infty)$ which is equivalent to $M(x)<0$ for $x>0$. Here we present the derivative of $M(x)$ as
\begin{align}\label{M_derivative}
M'(x)&=f(x)\left[1+\left(1-\frac{m_1(x)}{m_2(x)}\right)m_3(x)\right],
\end{align}
where
\begin{align}\label{cdef}
m_3(x)&=\frac{1-\beta}{2\alpha^{\beta}\beta}x^{-\beta}-\frac{1}{2}.
\end{align}
From (\ref{rdef}), (\ref{sdef}) and (\ref{cdef}), one obtains that $0<\frac{m_1(x)}{m_2(x)}<1$ and $m_3(x)>-\frac{1}{2}$. Consequently, $M'(x)$ (\ref{M_derivative}) satisfies
\begin{align}\label{Mder}
M'(x)\geq f(x)\left[1-\frac{1}{2}\left(1-\frac{m_1(x)}{m_2(x)}\right)\right]>\frac{1}{2}f(x)>0.
\end{align}
Combined with $\underset{x\to+\infty}{\lim} M(x)=\frac{-q_1}{(1-q_0-q_1)}\leq 0$, one concludes that $M(x)<\underset{x\to+\infty}{\lim} M(x)\leq 0$, and $G(x)$ is monotonically decreasing in $(0,+\infty)$. As a result, The maximum of $G(x)$ is obtained at $x^*=0$.
\subsubsection{$1<\beta\leq 2$ }\label{beta_12}
In this setup, we prove that $G(x)$ is quasiconcave and has only one stationary point in $(0,+\infty)$. First we introduce a point $x_1$ which is useful for the following analysis.
\begin{proposition}\label{prop_4}
Let $x_1=\frac{1}{\alpha}\left(\frac{\beta-1}{\beta}\right)^{\frac{1}{\beta}}$ and $\beta>1$, one has $M'(x_1)>0$.
\end{proposition}
\begin{proof}
The proof is postponed to \ref{prop4}.
\end{proof}
 Due to $\underset{x\to 0^+}{\lim}~m_1(x){x^{\beta-1}}=\frac{f(0)}{2\alpha^{\beta}\beta}$ and
\begin{align}\notag
1-\frac{m_1(x)}{m_2(x)}=\frac{1/(1-q_0-q_1)^2}{4 m_1(x)m_2^2(x)},
\end{align}
one has $\underset{x\to 0^+}{\lim} M'(x)=-\infty$. Combined with $M'(x_1)>0$ from Proposition \ref{prop_4}, one concludes that $M(x)$ has at least a stationary point $x_0\in(0,x_1)$ satisfying $M'(x_0)=0$. Utilizing $M'(x_0)=0$, from (\ref{M_derivative}) one has
\begin{align}\label{r_r1}
&\frac{m_1(x_0)}{m_2(x_0)}=1+\frac{1}{m_3(x_0)}>0.
\end{align}

Then we calculate the second derivative of $M(x)$ as
\begin{align}\label{M_2}
M''(x)&=\frac{f'(x)M'(x)}{f(x)}+f(x)\left(1-\frac{m_1(x)}{m_2(x)}\right)\times\notag\\
&\left[-\frac{\beta}{x}\left(m_3(x)+\frac{1}{2}\right)-\frac{f(x)c^2(x)}{m_2(x)}\left(1+\frac{m_1(x)}{m_2(x)}\right)\right],
\end{align}
where $f'(x)=-\alpha^{\beta}\beta x^{\beta-1}f(x)<0$ is the derivative of $f(x)$ $(x>0)$, and (\ref{M_2}) can be derived from (\ref{M_derivative}) via the basic product rule of differentiation formula $(uv)' = u'v +uv'$ and $m_1'(x)=f(x)m_3(x)$.
Substituting (\ref{rdef}), (\ref{r_r1}) and $M'(x_0)=0$ into (\ref{M_2}) yields
\begin{align}\label{M_2_x0}
&M''(x_0)=\frac{f(x_0)}{x_0}\left[\left(1+\frac{1}{m_3(x_0)}\right)(2-\beta)-\frac{\beta}{2m_3(x_0)}\right].
\end{align}
Given $1<\beta\leq2$, from (\ref{cdef}) one obtains $m_3(x)<-1/2$ for arbitrary $x\in(0+\infty)$. Therefore, $M''(x_0)>0$ due to $m_3(x_0)<0$ and $1+{1}/{m_3(x_0)}>0$ (\ref{r_r1}). Here we state the uniqueness of the stationary point $x_0$ in the following proposition.
\begin{proposition}\label{prop_5}
For a second-order differentiable univariate function $f(x) (a<x<b)$, and the stationary point $x_0$ such that $f'(x_0)=0$ satisfies $f''(x_0)>0$ ($f''(x_0)<0$), $f(x)$ is a quasiconvex (quasiconcave) function and $x_0$ is unique.
\end{proposition}
\begin{proof}
The proof is postponed to \ref{prop5}.
\end{proof}
According to Proposition \ref{prop_5}, $M(x)$ is quasiconvex and $x_0$ is unique. Then $M(x)$ is increasing in $(x_0,+\infty)$ and $M(x)<\underset{x\to +\infty}\lim M(x)$ for $x>x_0$. Combined with $\underset{x\to +\infty}\lim M(x)=\frac{-q_1}{1-q_0-q_1}$, one has $M(x)<\frac{-q_1}{1-q_0-q_1}\leq 0$ for $x>x_0$. Due to $\underset{x\to 0^+}\lim M(x)=\frac{q_0-q_1}{2(1-q_0-q_1)}>0$, $M(x_0)<0$ and $M(x)$ is decreasing in $(0,x_0)$, one concludes that there exists a point $x^*$ such that
\begin{align}\label{xstar}
M(x^*)=0,\quad  0<x^*<x_0,
\end{align}
and $M(x)>0$ in $(0,x^*)$ and $M(x)<0$ in $(x^*,x_0)$. Consequently, $G(x)$ is monotonically increasing in $(0,x^*)$ and monotonically decreasing in $(x^*,+\infty)$. $G(x)(x>0)$ is quasiconcave and achieves its maximum at the unique stationary point $x^*$.
\subsubsection{$\beta>2$}\label{beta_2}
In this setup, we prove that $G(x)$ is quasiconcave and has only one stationary point in $(0,+\infty)$.
For $\beta > 2$, revealing the quasiconcave property of $G(x)$ $(x > 0)$ is a little difficult than that of $1<\beta\leq 2$. Similar to Proposition \ref{prop_1}, we introduce a point $x_2$ as below.

\begin{proposition}\label{prop_6}
Let $x_2=\frac{1}{\alpha}\left(\frac{\beta-2}{2\beta}\right)^{\frac{1}{\beta}}$ and $\beta>2$, one has $M'(x_2)<0$.
\end{proposition}
\begin{proof}
The proof is postponed to \ref{prop6}.
\end{proof}
Due to $\underset{x\to 0^+}{\lim} M'(x)>0$, $M'(x_2)<0$, $M'(x_1)>0$ and $0<x_2<x_1$, we conclude that there exist at least two stationary points $x_0$ and $x'_0$ such that
\begin{subequations}\label{x0_x0}
\begin{align}
&M'(x_0)=0,\quad  0<x_0<x_2,\\
&M'(x'_0)=0,\quad  x_2<x'_0<x_1.
\end{align}
\end{subequations}
Substituting $m_3(x_0)<m_3(x_2)=\frac{2-\frac{3}{2}\beta}{\beta-2}<m_3(x'_0)<-\frac{1}{2}$ into (\ref{M_2_x0}) yields
\begin{subequations}
\begin{align}
&M''(x_0)=\frac{f(x_0)}{x_0}\left[2-\beta+(2-\frac{3}{2}\beta)\frac{1}{m_3(x_0)}\right]\notag\\
&<\frac{f(x_0)}{x_0}\left[2-\beta+(2-\frac{3}{2}\beta)\frac{\beta-2}{2-\frac{3}{2}\beta}\right]=0,\\
&M''(x'_0)=\frac{f(x'_0)}{x'_0}\left[2-\beta+(2-\frac{3}{2}\beta)\frac{1}{m_3(x'_0)}\right]\notag\\
&>\frac{f(x'_0)}{x'_0}\left[2-\beta+(2-\frac{3}{2}\beta)\frac{\beta-2}{2-\frac{3}{2}\beta}\right]=0
\end{align}
\end{subequations}
From Proposition \ref{prop_4}, in $(0,x_2)$ $M'(x)$ is quasiconcave and the stationary point $x_0$ is unique; in $(x_2,+\infty)$ $M(x)$ is quasiconvex and the stationary point $x'_0$ is unique. In addition, one can conclude that $M'(x)<0$ in $(x_0,x'_0)$, and $M'(x)>0$ in the rest intervals. From $M'(x)>0$ in $(0,x_0)$, one has $M(x_0)>\underset{x\to 0^+}{\lim}~M(x)=\frac{q_0-q_1}{2(1-q_0-q_1)}>0$. From $M'(x)>0$ in $(x'_0,+\infty)$, one has $M(x'_0)<\underset{x\to +\infty}{\lim}~M(x)=\frac{-q_1}{1-q_0-q_1}\leq 0$ and $M(x)<0$ for $x>x'_0$. Combined with $M'(x)<0$ in $(x_0,x'_0)$, one concludes that there exists a unique point $x^*$ such that
\begin{align}\label{xstar_2}
M(x^*)=0,\quad  x_0<x^*<x'_0.
\end{align}
Therefore, $M(x)>0$ in $(0,x^*)$ and $M(x)<0$ in $(x^*,x_0)$, and $G(x)$ is monotonically increasing in $(0,x^*)$ and monotonically decreasing in $(x^*,+\infty)$. $G(x)(x>0)$ is quasiconcave and achieves its maximum at the unique stationary point $x^*$.
\subsection{BSC}\label{BSC}
In this subsection, the optimal threshold in BSC is studied, which follows the results derived in the binary asymmetric channel case. Provided that $q_0=q_1=q$, from (\ref{x_-x}) it can be derived that $G(x)=G(-x)$ for arbitrary $x$. Therefore, the optimal threshold must be zero or pairs of opposite numbers. Now we prove that the optimal threshold in BSC is zero for $0<\beta\leq2$ and a pair of opposite numbers for $\beta<2$.
\subsubsection{$0<\beta\leq2$}\label{bsc_02}
For $0<\beta\leq 1$, following the similar derivation in section \ref{beta_01}, one concludes that $M'(x)>0$ (\ref{Mder}), and $M(x)<0$ due to $\underset{x\to+\infty}{\lim} M(x)=\frac{-q}{1-2q}\leq 0$.
For $1<\beta\leq 2$, following the similar derivation in section \ref{beta_12}, one concludes that $M(x)$ has at least a stationary point $x_0\in(0,+\infty)$ satisfying $M'(x_0)=0$ and $M''(x_0)>0$. From Proposition \ref{prop_5}, $M(x)$ is quasiconvex and $x_0$ is unique. Then $M(x)$ is decreasing in $(0,x_0)$ and increasing in $(x_0,+\infty)$. Due to $\underset{x\to +\infty}\lim M(x)=\frac{-q}{1-2q}\leq 0$ and $\underset{x\to 0^+}\lim M(x)=0$ (while in the setting $q_0>q_1$, $\underset{x\to 0^+}\lim M(x)=\frac{q_0-q_1}{2(1-q_0-q_1)}>0$), one has $M(x)<0$ for arbitrary $x\in(0,+\infty)$. Therefore, for $0<\beta\leq 2$, one concludes that $G(x)$ is decreasing in $x\in(0,+\infty)$ due to the same signs of $G'(x)$ and $M(x)$. Because of $G(x)=G(-x)$, $G(x)$ attains its maximum at zero.
\subsubsection{$\beta>2$}\label{bsc_2}
For $\beta>2$, $G(x)(x>0)$ is quasiconcave and achieves its maximum at the unique stationary point $x^*$. The proof is similar to that in section \ref{beta_2} except that $\underset{x\to 0^+}{\lim}~M(x)=0$ and $\underset{x\to +\infty}{\lim}~M(x)=\frac{-q}{1-2q}\leq 0$.
\subsection{Optimal Threshold Calculation}\label{optimal_beta_1}
In this subsection, first an upper bound for the optimal value is given. Then combing both asymmetric and symmetric cases, we propose a numerical algorithm to efficiently calculate the optimal value $x^*$ (\ref{Gx}) for arbitrary binary channel.
\subsubsection{Upper Bound}\label{upper}
In the arbitrary binary channel setup, we prove that $1/\alpha$ is an upper bound for the optimal point $x^*$.

Under the binary asymmetric channel, for $0<\beta\leq 1$, the optimal threshold zero proved in section \ref{beta_01} meets the bound. For $1<\beta\leq 2$, it is proved that $M(x)$ is quasiconvex and $x_0$ is the unique stationary point of $M(x)$ in section \ref{beta_12}. Hence we have $x_1>x_0$ from $M'(x_1)>M'(x_0)=0$. From (\ref{xstar}) we know that $x^*<x_0$. Therefore, we have
\begin{align}
x^*<x_0<x_1.
\end{align}
For $\beta>2$, from (\ref{x0_x0}) and (\ref{xstar_2}) in section \ref{beta_2}, we have
\begin{align}\label{x01alpha}
x^*<x'_0<x_1.
\end{align}
Therefore, $x^*$ is upper bounded by $x_1$ for $\beta>1$. In addition, $x_1$ is an increasing function with respect to $\beta$ and attains its maximum at $x_1|_{\beta=+\infty}=1/\alpha$, which results in
\begin{align}\label{bound2}
x^*<1/\alpha.
\end{align}

Under the BSC, for $0<\beta\leq 2$, the optimal threshold zero proved in section \ref{bsc_02} meets the bound. For $\beta>2$, similarly to the binary asymmetric channel case (\ref{x01alpha}), one has $x^*<x_0<x_1<1/\alpha$. As a result, $1/\alpha$ is an upper bound for the optimal threshold for arbitrary binary channel.

\subsubsection{Numerical Algorithm}
From section \ref{BAC} and \ref{BSC}, the optimal threshold is zero for $0<\beta\leq 1$ under binary asymmetric channel and for $0<\beta\leq 2$ under BSC. In other settings, the optimal threshold is non-zero. Utilizing the upper bound $1/\alpha$, we provide a numerical algorithm for efficient calculation of the non-zero optimal threshold, as shown in Algorithm \ref{algorithm_only}. Because the inequality constrained minimization problem $\underset{0<x<1/\alpha} {\operatorname{min}}~-G(x)$ has a unique stationary point, and the first-order descent methods converge to a stationary point, a gradient descent algorithm is guaranteed to find the global optimum.
\begin{algorithm}
\protect\caption{Gradient Descent Algorithm}\label{algorithm_only}
\begin{enumerate}
\item Initialize $k=0$ and $x_{k}\in(0,1/\alpha)$.
\item Set $\triangle x_k=-G'(x_{k})$ (\ref{G_derivative}).
\item Choose a step size $t$ via backtracking linear search, satisfying $G(x_k+t\triangle x_k)\leq G(x)+0.4 t G'(x_k)\triangle x$ and $x_k+t\triangle x_k\in(0,1/\alpha)$.
\item Update $x_{k+1}=x_k+t\triangle x_k$.
\item Set $k=k+1$ and return to step 2 until the stopping criterion $|G'(x_k)|<10^{-5}\alpha^3$ is satisfied.
\end{enumerate}
\end{algorithm}
\section{Numerical Simulations}\label{simulation}
In section \ref{Q_design}, it is proven that $x^*=0$ for $0<\beta\leq 1$, and $G(x)(x>0)$ is quasiconcave for $\beta >1$ in the asymmetry binary channels $q_0>q_1$ and $1-q_0-q_1>0$. Utilizing the quasiconcavity, the numerical algorithm is conducted to obtain the maximum of $G(x)$, and the effectiveness of the corresponding optimal threshold is verified via numerical simulations.

For the first experiment, we use gradient descent algorithm to find the optimal threshold normalized by the scale parameter $\alpha^{-1}$ of the GGN. The results are presented in Fig. \ref{optimum_beta}.
\begin{figure}[h!t]
\centering
\includegraphics[width=80mm]{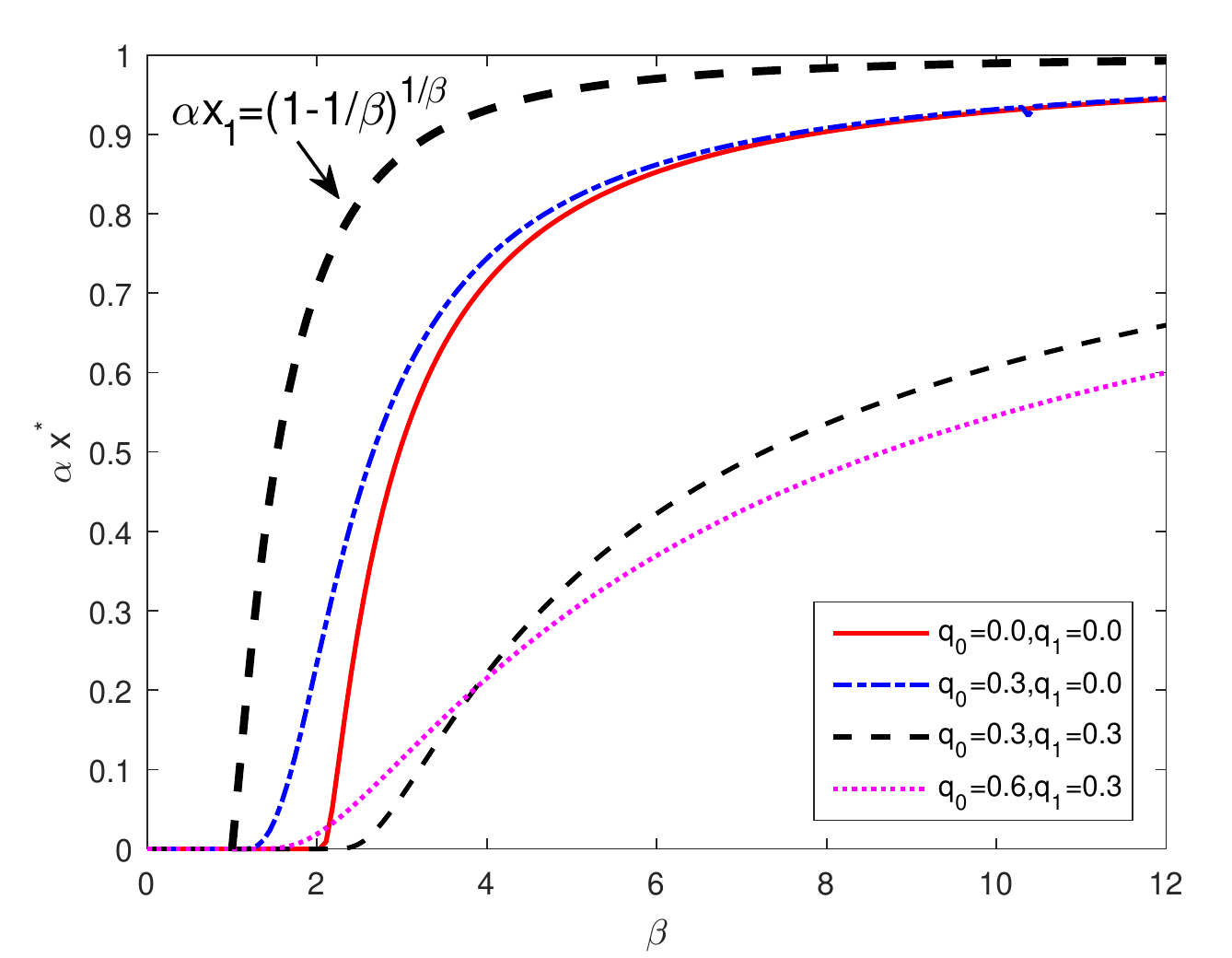}
\caption{The relationship between the normalized $\alpha x^*$ and $\beta$ under different flipping probabilities $(q_0,q_1)$.}
\label{optimum_beta}
\end{figure}
It shows that for $0<\beta\leq 1$, the optimal threshold is zero; for $1<\beta\leq 2$, the optimal threshold is zero under $q_0=q_1$, and non-zero under $q_0\not =q_1$; for $\beta> 2$, the optimal threshold is non-zero. In addition, for arbitrary flipping probabilities, the optimal threshold increases with $\beta$ and is upper bounded by $x_1=\frac{1}{\alpha}(1-\frac{1}{\beta})^{\frac{1}{\beta}}$.

For the second experiment, the effectiveness of quantizer thresholds design is verified. In TABLE \ref{x_star}, $x^*$ under different $\beta$ is calculated by Algorithm \ref{algorithm_only}. The corresponding optimal threshold $\tau^*$ is $-x^*$. Parameters are set as follows: $\alpha=1$, $\theta = 0.0661$, $q_0=0.7$, $q_1=0$, $N=2000$, $K = 1$, $h_{ij}=1,~\forall~ i,j$, the number of Monte Carlo trials is $2000$. The receiver operating characteristic (ROC) curves, i.e., the detection probability $P_{\rm D}$ versus the false alarm probability $P_{\rm FA}$, are presented in Fig. \ref{Pd_Pfa}. We have noticed that the ROCs of the Rao test are similar to those of the GLRT. To present the results clearly, we do not plot the ROCs of the Rao test in this experiment.
\begin{table}[h!t]
    \begin{center}
\caption{The values of $x^*$ under different $\beta$ with flipping probabilities $(q_0,q_1)=(0.7,0)$}.\label{x_star}
    \begin{tabular}{|l|c|c|c|c|}
           \hline{$\beta$ }  &   $1.5$ & $2$ & $4$ & $8$ \\
            \hline$x^*$  &  $0.1200$ & $0.3682$ & $0.7727$ &$0.9130$ \\ \hline
       \end{tabular}
    \end{center}
\end{table}
\begin{figure}[h!t]
\centering
\includegraphics[width=140mm]{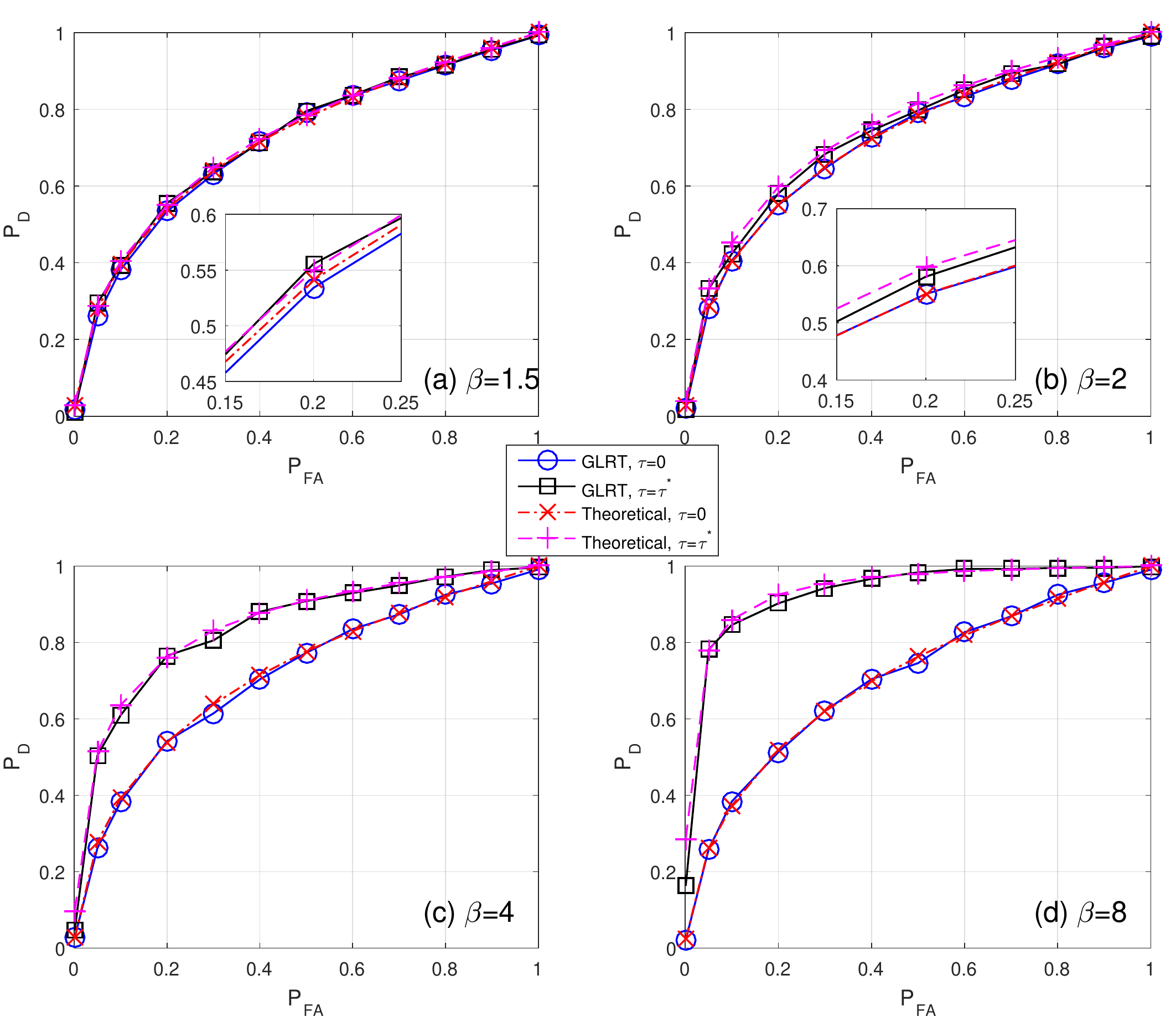}
\caption{The ROC curve under $\beta>1$. The flipping probabilities are $q_0=0.7$ and $q_1=0$.}
\label{Pd_Pfa}
\end{figure}

From TABLE \ref{x_star} and Fig. \ref{Pd_Pfa}, one obtains that under certain flipping probabilities $q_0\not=q_1$ and $\beta>1$, the performance of GLRT is improved by using the optimal threshold. When $\beta$ is small, the gain of the quantizer design with respect to the zero-threshold is negligible because the optimal threshold is still close to zero. As $\beta$ increases, the detection performance of the designed quantizer improves significantly compared to that utilizing the zero-threshold.


For the third experiment, we detect an one dimensional acoustic field under the ship transit noise \cite{Hodgkiss}. Let $h_{ij}=\sin(kx_i-\omega t_j)$ denote the unit response of the acoustic field at position $x_i$ and time instant $t_j$, $k$ is the wave number and $\omega$ is the angular frequency. In $25^{\circ}$C seawater (in which the sound speed is about $1500$ m/s), $50$ sensors are equispaced in $100$ m to test for the presence of a weak sound wave whose amplitude is $0.1$ Pa and frequency is $200$ Hz. For sensors, the sampling frequency is $5000$ Hz, and the sampling time is $0.1$ s. Accordingly, parameters are set as follows: $\theta=0.05$, $\alpha=1$, $q_0=0.3$, $q_1=0$, $N=50$, $K=50$, $x_i=2i$, $t_j=j/500$, $k=400\pi /1500\approx0.8378$, $\omega=400\pi\approx1257$, and $h_{ij}=\sin(1.676i-2.514j)$. In addition, the GGN with $\beta=2.779$ represents the ship transit noise \cite{Banerjee}. The number of Monte Carlo trials is $10^3$, and the ROC curves are presented in Fig. \ref{wave_vector}. It can be seen that the ROCs of the Rao test are almost the same as those of the GLRT. Compared to using the suboptimal zero-threshold, utilizing the optimal threshold improves the performances of the GLRT and Rao detectors.
\begin{figure}[h!t]
\centering
\includegraphics[width=80mm]{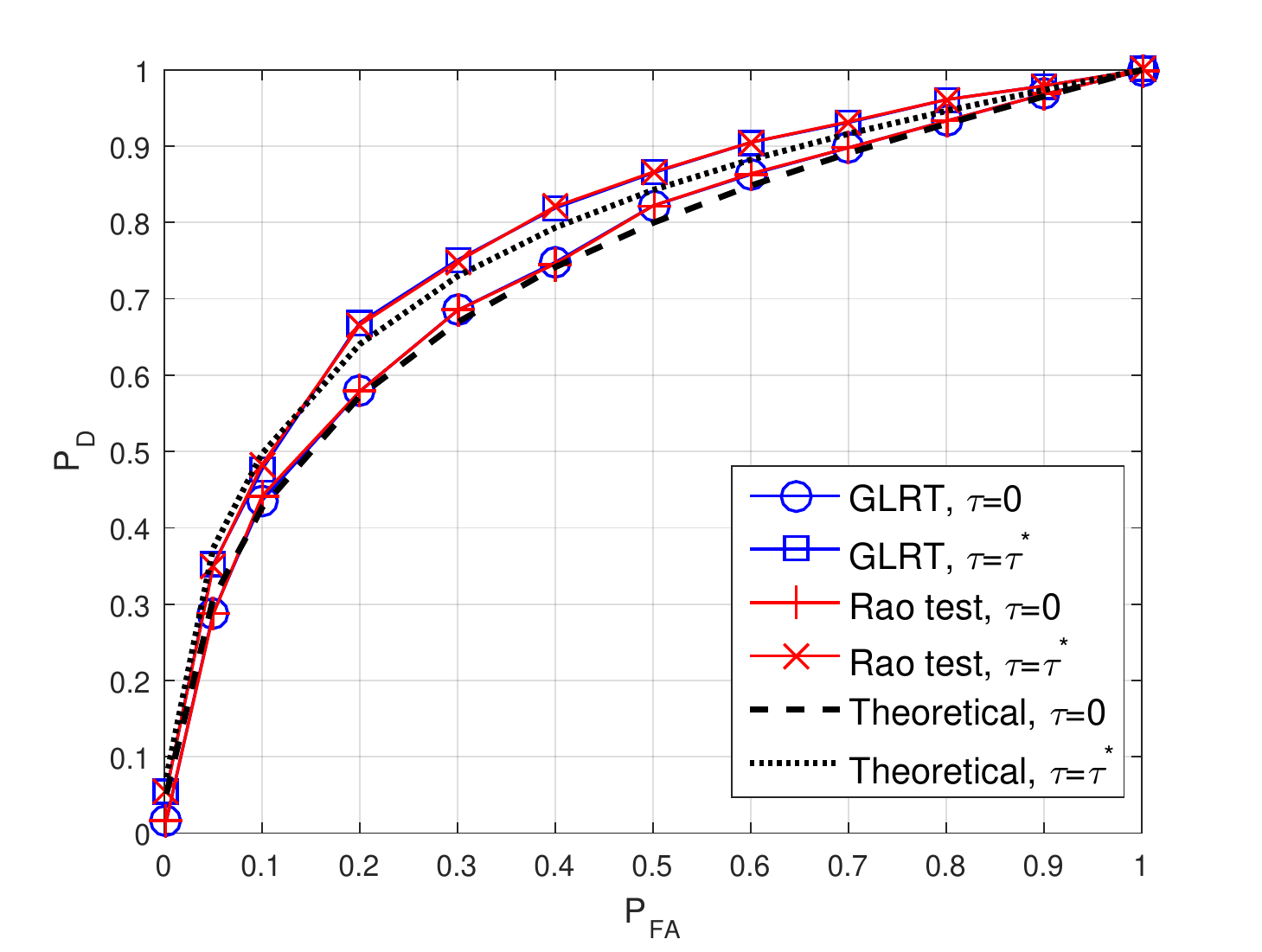}
\caption{The ROC curve for detecting the acoustic wave field under the ship transit noise environments.}
\label{wave_vector}
\end{figure}
\section{Conclusion}\label{con}
Provided that the noise obeys the generalized Gaussian distribution, it is shown that the optimal threshold depends on the value of shape $\beta$ critically. For $0<\beta\leq 1$, the optimal threshold is zero in both binary symmetric and asymmetric channels. For $1<\beta\leq 2$, the optimal threshold is zero in the BSC, while it is non-zero and unique in the binary asymmetric channel. For $\beta>2$, in the BSC, there exist two non-zero solutions which are opposite numbers corresponding to optimal thresholds, while in the binary asymmetric channel the optimal threshold is non-zero and unique. Next, for the cases of non-zero optimal thresholds, we prove that maximizing the non-central parameter can be solved efficiently via numerical algorithm. Finally, the effectiveness of the optimal threshold is verified in numerical experiments, and the gain of using the designed threshold becomes larger as the shape parameter $\beta$ increases.
\section{Acknowledgement}\label{ack}
This work is supported by the Zhejiang Provincial Natural Science Foundation of China under grant No. LQ18F010001 and the Fundamental Research Funds for the Central
Universities under Grant No. 2017QNA4042.
\begin{appendix}
\section{ Proof of Propositions}\label{appendix}
\subsection{Proposition 1}\label{prop1}
\begin{proof}
$\forall$ $\alpha>0$, $\beta>0$, $0\leq q_0 \leq1$ and $0\leq q_1 \leq1$, the equalities
\begin{align}\label{lemma_xab}
&G(x,q_0,q_1)=G(-x,q_1,q_0)=G(x,1-q_0,1-q_1)\notag\\
&=G(-x,1-q_1,1-q_0).
\end{align}
hold, due to $f(x)=f(-x)$ and $F(x)+F(-x)=1$. Let $(q_0=q_a, q_1=q_b)$ satisfy $q_0>q_1$ and $1-q_0-q_1>0$, and $x^*$ denote the value of $x$ which attains the maximum of $G(x,q_0,q_1)$. According to (\ref{lemma_xab}), $G(x^*,q_a,q_b)=G(x,q_b,q_a)|_{x=-x^*}=G(x^*,1-q_a,1-q_b)=G(x,1-q_b,1-q_a)|_{x=-x^*}\geq G(x,q_a,q_b)=G(-x,q_b,q_a)=G(x,1-q_a,1-q_b)=G(-x,1-q_b,1-q_a)$. The maximums of $G(x,q_b,q_a)$, $G(x,1-q_a,1-q_b)$ and $G(x,1-q_b,1-q_a)$ are obtained at $-x^*$, $x^*$ and $-x^*$, corresponding to the cases that $q_0<q_1$ \& $q_0+q_1<1$, $q_0<q_1$ \& $1-q_0-q_1<0$ and $q_0>q_1$ \& $1-q_0-q_1<0$, respectively. As a consequence, we conclude that the maximum of $G(x)$ in the case that $q_0<q_1$ or $1-q_0-q_1<0$ can be transformed into the case that $q_0>q_1$ and $1-q_0-q_1>0$.

Given that $q_0>q_1$ and $1-q_0-q_1>0$, for $x>0$, we have
\begin{align}\label{x_-x}
&G(x)-G(-x)=\frac{1}{\frac{1}{4}-\left[(1-q_0-q_1)F(x)-\frac{1}{2}+q_1\right]^2}\notag\\
&\times\frac{2f^2(x)[F(x)-\frac{1}{2}](1-q_0-q_1)(q_0-q_1)}{\frac{1}{4}-\left[(1-q_0-q_1)F(x)-\frac{1}{2}+q_0\right]^2}.
\end{align}
Utilizing $\frac{1}{2}<F(x)<1$, we have
\begin{subequations}
\begin{align}
\frac{q_1-q_0}{2}<(1-q_0-q_1)F(x)-\frac{1}{2}+q_1<\frac{1-2q_0}{2},\notag\\
\frac{q_0-q_1}{2}<(1-q_0-q_1)F(x)-\frac{1}{2}+q_0<\frac{1-2q_1}{2},\notag
\end{align}
\end{subequations}
which guarantee the inequalities
\begin{subequations}
\begin{align}
\left|(1-q_0-q_1)F(x)-\frac{1}{2}+q_1\right|<\frac{1}{2},\notag\\
\left|(1-q_0-q_1)F(x)-\frac{1}{2}+q_0\right|<\frac{1}{2}.\notag
\end{align}
\end{subequations}
Therefore, the denominators of both terms in (\ref{x_-x}) are positive and $G(x)-G(-x)>0$. Because $x = 0$ is also a feasible point of $G(x)$, the optimal point $x^*$ is either equal to zero or in the interval $(0,+\infty)$.
\end{proof}
\subsection{Proposition \ref{prop_4}}\label{prop4}
\begin{proof}
For $\beta>1$, from (\ref{cdef}) we have
\begin{align}\label{cx1_1}
m_3(x_1)=-1.
\end{align}
Substituting (\ref{cx1_1}) into (\ref{M_derivative}) yields
\begin{align}
M'(x_1)=\frac{f(x_1)m_1(x_1)}{m_2(x_1)}>0.
\end{align}
\end{proof}
\subsection{Proposition \ref{prop_5}}\label{prop5}
\begin{proof}
The proof refers to \cite[p. 101]{boyd}. From $f''(x_0)>0$ ($f''(x_0)<0$), we know that whenever the function $f'(x)$ crosses the value $0$, it is strictly increasing (decreasing). Therefore $f'(x)$ can cross the value $0$ at most once. It follows that $f'(x)<0$ for $a<x<x_0$ and $f'(x)>0$ for $x_0<x<b$ ($f'(x)<0$ for $a<x<x_0$ and $f'(x)>0$ for $x_0<x<b$). This shows that $f(x)$ is quasiconvex (quasiconcave) and the stationary point $f'(x_0)=0$ is unique.
\end{proof}
\subsection{Proposition \ref{prop_6}}\label{prop6}
\begin{proof}
For $\beta>2$, from (\ref{M_derivative}) and $m_3(x_2)=-(3/2 +1/(\beta-2)<-3/2$, we have
\begin{align}
\frac{M'(x_2)}{f(x_2)}\leq1+\left(1-\frac{m_1(x_2)}{\sqrt{\frac{1}{4}+m_1^2(x_2)}}\right)m_3(x_2),
\end{align}
in which the condition for equality is $q_0=q_1=0$ or $1$.
To prove that $M'(x_2)<0$ for arbitrary $(q_0,q_1)$ is equivalent to prove that
\begin{align}
1+\left(1-\frac{m_1(x_2)}{\sqrt{\frac{1}{4}+m_1^2(x_2)}}\right)m_3(x_2)<0,
\end{align}
which can be simplified as
\begin{align}\label{complex1}
\frac{1}{m_1^2(x_2)}>4\left[\left(\frac{1}{1+\frac{1}{m_3(x_2)}}\right)^2-1\right].
\end{align}
Substituting (\ref{cdef}) and (\ref{rdef}) into (\ref{complex1}) yields
\begin{align}\label{complex2}
\Gamma^2(1/\beta)>\left[e^{\frac{2-\beta}{2\beta}}\left(\frac{2\beta}{\beta-2}\right)^{\frac{\beta-1}{\beta}}\right]^2\frac{2(\beta-1)(\beta-2)}{\beta^2},
\end{align}
whose logarithm is
\begin{align}\label{complex3}
2\ln\Gamma(1/\beta)&>\frac{2}{\beta}-1+\left(3-\frac{2}{\beta}\right)\ln2-\frac{2}{\beta}\ln\beta\notag\\
&+\left(\frac{2}{\beta}-1\right)\ln(\beta-2)+\ln(\beta-1).
\end{align}
Let $t=1/\beta$, then $0<t<\frac{1}{2}$ due to $\beta>2$. Utilizing $\Gamma(x+1)=x\Gamma(x) (x>0)$, (\ref{complex3}) can be transformed as
\begin{align}
2\ln\Gamma(t+1)&>2\ln t+2t-1+\left(3-2t\right)\ln2+\notag\\
&(2t-1)\ln(1-2t)+\ln(1-t)\triangleq Q(t).
\end{align}
According to \cite{nature}, the minimum of $\Gamma(t+1)(0<t<1/2)$ is obtained at $t=0.461$. Now, we prove that $Q(t)<2\ln\Gamma(1.461)= -0.2430$ for $0<t<1/2$. The first and second order derivatives of $Q(t)$ are
\begin{subequations}\label{complex5}
\begin{align}
Q'(t)&=4-2\ln2+\frac{2}{t}+\frac{1}{t-1}+2\ln(1-2t),\\
Q''(t)&=-\frac{2}{t^2}-\frac{4}{1-2t}-\frac{1}{(1-t)^2}.
\end{align}
\end{subequations}
Given $0<t<1/2$, $Q''(t)<0$ and $Q(t)$ is concave.
We use the \emph{MATLAB} \emph{fminunc} function and obtain the maximum of $Q(t)$, achieved at $t=0.4609$ (very near the optimal point $t=0.461$ of $\Gamma(t+1)(0<t<1/2)$). Since $Q(t)\leq Q(0.4609)=-0.60542<-0.2430=2\ln\Gamma(1.461)\leq 2\ln\Gamma(t+1)$, the proposition is proved.
\end{proof}
\end{appendix}


\end{document}